\documentclass[12pt]{article}

\usepackage[square,numbers]{natbib}
\usepackage{hyperref}
\hypersetup{
	colorlinks=true,
	linkcolor=blue,
	filecolor=black,      
	urlcolor=blue,
	citecolor=blue,
}

\usepackage{amsfonts}
\usepackage{graphicx}
\usepackage{enumitem}
\usepackage{centernot}

\usepackage{amsmath}
\usepackage{dsfont} 
\usepackage{amsthm}
\usepackage{color, soul}
\usepackage{tensor}   

\usepackage[normalem]{ulem}

\newtheorem{Theo}{Theorem}[section]

\newtheorem{Prop}[Theo]{Proposition}

\theoremstyle{definition} 
\theoremstyle{definition} 
\theoremstyle{definition} 

\newcommand{\be}{\begin{equation}}
\newcommand{\ee}{\end{equation}}
\newcommand{\bea}{\begin{equation}\begin{aligned}}
\newcommand{\eea}{\end{aligned}\end{equation}}

\newcommand*{\vNAlg}{\overline{\mathcal{A}}}
\newcommand{\Alg}{\mathcal{A}}
\newcommand{\Hil}{\mathcal{H}}
\newcommand{\Zen}{\mathcal{Z}}
\newcommand{\Trans}{\mathcal{T}}

\newcommand{\wlim}{\underset{\!\!n\to\infty}{\mathrm{w\!-\!lim}}\,}

\title{Generalized criteria of symmetry breaking. \\
	   A strategy for quantum time crystals}
\author{Carlo Heissenberg$^1$ and Franco Strocchi$^2$\\
	{\footnotesize $^1$ carlo.heissenberg@sns.it, Scuola Normale Superiore and INFN, Pisa, Italy}\\
	{\footnotesize $^2$ Physics Department, University of Pisa, Pisa, Italy}
}

\date{}
\begin{document}

\maketitle

\begin{abstract} 
	{The aim of this paper is to propose a criterion of spontaneous symmetry breaking that makes reference to the properties of pure phases defined by a translationally invariant state. By avoiding any reference to the ground state, at the basis of the standard approach, this criterion applies to a wider class of models. 
	An interesting application is the breaking of time translations.
	Indeed, we discuss explicit theoretical models which exhibit the prototypical features of quantum time crystals, without the need of a time-dependent Hamiltonian.}
\end{abstract}

\noindent
Keywords: Symmetry Breaking Criteria; Pure Phases; Breaking of Time Translations; Quantum Time Crystals

\section*{Introduction}

{From first principles, a symmetry of the local observables is spontaneously broken whenever it cannot be implemented by a unitary operator, \emph{i.e.} it is not a Wigner symmetry, in a given (pure) phase. This general definition, however, is operationally not easy to check in explicit cases, and in fact more practical criteria of symmetry breaking have been proposed in the literature, typically in relation to the field of application. The standard criterion relies on the non-invariance of the ground or equilibrium state, while other criteria have been proposed to cover the otherwise precluded case of the breaking of time translations \cite{Watanabe, Else, KhemaniPRB94085, Khemani}. 

The main purpose of this paper is to propose a criterion of spontaneous symmetry breaking which directly involves the properties of a pure phase defined by a translationally invariant state. In this way, in our opinion, one has to control properties which are relatively easy to check and have a simple operational meaning, also in view of the experimental preparation.

As it is well known from general theoretical treatments \cite{Haag, BratRobI, BratRobII}, the implementation of the symmetry in a mixed phase does not imply the absence of a symmetry breaking order parameter which might appear in the pure phase decomposition. For this reason, we directly consider pure phases. Another crucial ingredient for the criterion of spontaneous breaking of symmetries that commute with space translations is the uniqueness of the translationally invariant state in the given phase. We are therefore lead to propose a criterion for the breaking of an internal symmetry in phases defined by a translationally invariant state that satisfies the cluster property.

Even if this setting includes the special case of a pure phase defined by a translationally invariant ground or equilibrium state, it is much more general and in particular provides an escape to the no-go theorem \cite{Bruno, Watanabe, Khemani} for the breaking of time translations and the existence of time crystals in phases of Hamiltonian systems defined by a ground or equilibrium state.
}

{
In Section \ref{sec: stbb}, after reviewing the general theoretical framework, we introduce our generalized criterion of spontaneous symmetry breaking. In Section \ref{sec: constructive!}, we consider its possible use for the breaking of time translation symmetry. The comparison with other strategies recently proposed in the literature is then discussed.
} 

{In Section \ref{sec: models} we present explicit models described by \emph{time-independent} Hamiltonians, which exhibit the spontaneous breaking of time translations. In particular we discuss models with a residual discrete subgroup of time translations; such a periodicity in time realizes the characteristic property of quantum time crystals.}

In the first example, this mechanism is realized thanks to the occurrence of a nontrivial topology, which gives rise to different pure phases, with no need for the thermodynamic limit. 
The second example is a spin model with short-range interactions and  coupling to a time-independent external field. In this case the thermodynamic limit and the pure phases defined by a translationally invariant state play an essential role. In our opinion such models, even without aiming to have a direct experimental realization, qualify as basic theoretical prototypes of quantum time crystals, achieved without the introduction of an external, periodic drive, which would spoil the one-parameter group property of the dynamics \cite{Yao}. The strategy of looking for a realization of quantum time crystals by spin models, proposed in \cite{FSCH_QTC} {(see also \cite{HeisenbergQTC})}, has been realized in a clever experimental setup \cite{Autti-Volovik}, where the external periodic drive is eventually sent to zero; the experiment detects spin oscillations surviving the limit in which the external drive is removed, showing the possibility of quantum time crystals. 
Another interesting proposal in this direction has been put forward in \cite{Fazio-Floquet}.
Finally, we discuss an improved version of the Wilczek model \cite{Wilczek}, originally proposed as an example of quantum time crystal.

{
Let us stress once more that, although most recent developments concerning the realization of quantum time crystals have mainly resorted to the introduction of a fine-tuned external periodic drive in the spirit of Floquet time crystals \cite{Yao,Khemani, timecrystalsREVIEW, Autti-Volovik, Federica-clock, MatusSacha-fractional, Federica-gauge}, we only consider models with time-independent Hamiltonians. 
}

In all {our} examples, the dynamics is described by a one-parameter group of automorphisms of the algebra of observables, which, however, cannot be implemented by a family of unitary operators in the Hilbert space of the system, a characteristic feature of any broken symmetry 
The time evolution of the observables is thus well-defined, maps the algebra of observables into itself, and in this way defines the time evolution of their expectations, but this does not correspond to a dual interpretation in terms of transformations of the states.
This means that the dynamics is well-represented in the Heisenberg picture but not in the Schr\"odinger picture.

\section{Symmetry Breaking and its Operational \\Criteria}\label{sec: stbb}
In order to discuss criteria of symmetry breaking, generalizing the standard identification with noninvariant ground states, we start by recalling the precise, general definition of spontaneous symmetry breaking.
 
 \subsection{Spontaneous symmetry breaking}
 \label{sec: Wsbs}
 A transformation $\beta$ of the canonical variables, or more generally of the algebra of observables, defines a (algebraic) symmetry if it preserves all the algebraic relations (\emph{e.g.} the canonical commutation relations). 

 In the following, we shall denote by $\Alg$ the algebra generated by the canonical variables or by the observables; technically it is convenient to consider the $C^\ast$-algebra generated by them.
 
 A state $\omega$ of the system is characterized by the set $\{\omega(A),A\in\Alg\}$ of its expectations; by a general result, each state defines a (GNS \cite{GNoriginal, S_original}) \emph{representation}, \emph{i.e.} a realization, of $\Alg$.
 
 In order to understand the phenomenon of Spontaneous Symmetry Breaking (SSB), a crucial point is that a given physical system, described by its algebra $\Alg$, may admit different disjoint realizations, or briefly \emph{phases},
 with the characteristic property of being stable under the action of $\Alg$.
 In the case of the algebra of observables, this means that no experimental device or physical process may induce transitions between different phases.
 
 This structures significantly emerges in the description of infinitely extended systems, as displayed by many-body systems in the thermodynamic limit or by relativistic quantum fields.
 
 Furthermore, for the description of SSB a crucial role is played by the so-called \textit{pure phases}. At zero temperature these are defined as the irreducible representations of $\Alg$. At finite temperature irreducibility does not hold and must be replaced by a less stringent condition, {satisfied by the standard thermodynamical pure phases (at equilibrium): one has to consider \emph{factorial} representations, namely those in which every element of the center $\mathcal Z$ of the observables is represented by a multiple of the identity.
 (Technically, we recall that a given representation naturally associates to $\Alg$ its weak closure $\vNAlg$, whose center $\Zen$ is the set of operators of $\vNAlg$ {that} commute with $\vNAlg$.)}
 
 For finitely extended systems, factoriality automatically provides a characterization of inequivalent representations of $\Alg$.
 For infinitely extended systems this property takes into account that average observables (equivalently, their values at space infinity) commute with $\vNAlg$ and therefore belong to $\Zen$; since they describe macroscopic quantities, they must attain sharply defined values according to the standard thermodynamical characterization of a pure phase. Factoriality can therefore be taken as characterization of pure phases.
 
 It is important to stress that a factorial representation of the algebra of observables describes a stable realization of the system, since any physically realizable operation (or observable) is essentially localized and therefore cannot change the boundary conditions, equivalently the values of the observables, at space infinity. It is also worthwhile to remark that a pure phase need not contain (or be defined by) an equilibrium or ground state; in fact one can explicitly construct pure phases defined by a stationary state not invariant under time translations. 
 This will play a crucial role for the breaking of time translations (see Section \ref{sec: constructive!}).

Another argument in favor of considering pure phases is their stability against small perturbations of the dynamics, as clearly emphasized by Weinberg {\cite[pp. 166-167]{WeinbergQFT2}}. 
 
 After these premises, we have the following general definition of SSB:
 \\
 \textit{A symmetry $\beta$ is \textbf{spontaneously broken} in a pure phase $\Gamma$ when it does not leave $\Gamma$ stable, \emph{i.e.} if there exists $\omega\in\Gamma$ such that $\omega_\beta$, defined by $\omega_\beta(A)=\omega(\beta^{-1}(A))$ for all $A\in\Alg$, does not belong to $\Gamma$.
 Equivalently, when $\beta$ cannot be realized as a unitary operator on the corresponding Hilbert space $\Hil_\Gamma$. The alternative corresponds to the symmetry being \textbf{unbroken} (Wigner symmetry).} 
 \\
 From an operational point of view, SSB is signaled by the noninvariance of some transition amplitude. 
 
 Clearly, if, for a given temperature, $\Alg$ admits only one pure phase, there is no room for SSB. Typical examples are the atomic systems (more generally, quantum-mechanical systems with a finite number of degrees of freedom described by the standard canonical variables), as a consequence of the Stone-von Neumann uniqueness theorem \cite{Stone, vonNeumann}. 
 On the other hand, in general, infinitely extended systems admit different, inequivalent realizations corresponding to different pure phases.

 \subsection{Generalized criteria of symmetry breaking}\label{sec: Ground}
 
The general characterization of SSB given above does not qualify as an easily checkable criterion for its occurrence.

The standard criterion identifies SSB with the noninvariance of the ground state that defines the corresponding phase.
{For systems in the thermodynamic limit}, we propose an alternative, more general and useful criterion of SSB, under the following conditions: 
\begin{enumerate}[label=(\roman*)]
	\item space translations $\mathrm{\mathbf{T}}$ are symmetries of $\Alg$;
	\item $\beta$ commutes $\mathrm{\mathbf{T}}$;
	\item one considers a pure phase $\Gamma_\omega$ (not necessarily at zero temperature) defined  
	by a state $\omega$ invariant under a subgroup $\Trans$ of $\mathrm{\mathbf{T}}$;
	\item $\Alg$ satisfies $\Trans$--asymptotic abelianess, namely 
	\begin{equation}
	\lim_{n\to\infty} \left[T^n(A), B\right]=0\qquad \forall A,B\in \Alg 
	\end{equation}
	for all nontrivial $T\in\Trans$ (at least as weak limit).
\end{enumerate}

Properties (i) and (ii) are in general satisfied for SSB considered in many-body theories and quantum field theories, while (iv) is the mildest requirement of localization, always satisfied in all algebras of observables used to describe condensed matter systems and elementary particles.

It is important to stress that the state $\omega$ in (iii) does not have to be a state invariant under time translations, \emph{i.e.} it need not be the ground/vacuum state. 
The crucial property shared by $\omega$ is that, as a consequence of factoriality ad asymptotic abelianess, it satisfies the cluster property \cite[pp.126 and 134]{Haag}
\be\label{cluster_T}
\lim_{n\to\infty} \left[
\omega\left(AT^n(B)\right)-\omega(A)\omega(T^n(B))
\right]=0\,.
\ee
Furthermore, since $\omega$ is invariant under translations,
\eqref{cluster_T} is equivalent to $\omega$ being the \emph{unique translation invariant state} in $\Gamma_\omega$.

Therefore the following criterion has a much wider range of application, compared to the standard one 

\begin{Prop}[General criterion of SSB]\label{criterio_t}
Under the conditions {$\mathrm{(i)}$---$\mathrm{(iv)}$}, the symmetry $\beta$ is spontaneously broken in $\Gamma_\omega$ if and only if there exists an $A\in\Alg$ such that
\be\label{eq: crit_inv}
\omega\left(\beta(A)\right) \neq \omega(A)\,,
\ee
\emph{i.e.} $\omega$ gives rise to a symmetry breaking order parameter.
\end{Prop}

\begin{proof}
	By a general result (see \emph{e.g.} \cite[p.71]{FS_SSB}), the invariance of $\omega$ under $\Trans$ implies that $\Trans$ is unbroken, and therefore implemented by unitary operators $U_T$, for all $T\in\Trans$. The same holds for $\beta$ if $\omega\left(\beta(A)\right)=\omega(A)$ for all $A\in\Alg$.
	
	Conversely, if $\beta$ is unbroken, \emph{i.e.} implemented by a unitary operator $U_\beta$ in the Hilbert space corresponding to $\Gamma_\omega$, the state $\omega_\beta$ 
	is invariant under $\Trans$, because so is $\omega$ and $\beta$ commutes with $\Trans$. 
	On the other hand, since $\omega$ defines a pure phase and is invariant under translations,
	it is the unique translation invariant state. Then $\omega_\beta=\omega$.
\end{proof}

In infinitely extended systems, under conditions (i)---(iv), the characterization of SSB by equation \eqref{eq: crit_inv} is equivalent to the existence of a \textit{long-range order}, as a consequence of the cluster property:
\begin{equation}
\begin{aligned}
\lim_{n\to\infty}\omega\left(T^n(\Delta A)B\right)=\omega(\Delta A)\omega(B)\neq 0,
\end{aligned}
\end{equation}
where $\Delta A \equiv \beta(A)-A$.

However, the noninvariance under $\beta$ of a generic state in the given (pure) phase does not say anything about $\beta$ being spontaneously broken: this is why eq. \eqref{eq: crit_inv} involves a special state $\omega$ invariant under a subgroup of translations and defining a pure phase (condition (iii)). 

An alternative criterion of SSB may be formulated for a pure phase $\Gamma$ without making reference to the above conditions (i)---(iv), by exploiting the transformation properties of the center $\mathcal Z$ (defined by weak limits of elements of $\Alg$) under the given symmetry.
\begin{Prop}\label{Criterio_centro}
	Given a pure phase $\Gamma$, if there exists $z\in\Zen$, given by
	\be\label{z_1c}
	z=\wlim A_n
	\ee 
	for $A_n\in\Alg$, and if, for some $\omega\in\Gamma$, 
	\be\label{z_2c}
	\omega(\beta(A_n))-\omega(A_n) \centernot\longrightarrow 0
	\ee
	as $n\to\infty$, then $\beta$ is spontaneously broken in $\Gamma$.
\end{Prop}
\begin{proof}
	If $\beta$ is unbroken, it is implementable by a unitary operator $U_\beta$ and therefore $\beta(A_n)$ is weakly convergent; indeed, for any state vector $\Phi$, letting $\Phi_\beta=U_\beta \Phi$,
	$$
	(\Phi, \beta(A_n) \Phi)=(U_\beta \Phi, A_n U_\beta \Phi)=(\Phi_\beta, A_n \Phi_\beta).
	$$
	Hence, $\beta$ can be extended to the center by weak continuity and
	$$
	\beta(z)=\beta\big(\wlim A_n\big)=U_\beta^\ast z U_\beta.
	$$
	In a factorial representation, the center is represented by multiples of the identity and therefore $\beta(z)=z$.
\end{proof}

{As above, this criterion is more general than the conventional one}

The symmetry breaking criterion given by Proposition \ref{Criterio_centro} becomes particularly simple when $z\in\Alg$, no weak limit being involved, and it also works for finite systems (an example is provided by the particle on a circle, as discussed below).

In a pure phase, defined by a state $\omega$ satisfying (iii), a useful class of central elements is available,  
because by asymptotic abelianess the following limit exist
\be
A_{\text{av}}
=\wlim \frac{1}{N}  \sum_{n=0}^N T^n(A)
=\wlim T^n(A) = \omega(A) \mathds 1\,,
\ee
where
$
T\in\Trans
$ and $A\in\Alg$, and have the meaning of variables at infinity or macroscopic observables.

\section{Spontaneous Breaking of Time Translations: a Constructive Strategy}\label{sec: constructive!}

{
Since the expectation of any operator on the
ground (or equilibrium) state $\omega_0$ is by definition invariant under time evolution $\alpha_t$, and therefore cannot signal the breaking of time translations, a way out has been to devise criteria which consider the properties of two-point functions $\langle A(t)B(0)\rangle$.
}

{In \cite{Watanabe}, the authors propose to define time crystals
by ground-state or equilibrium correlation functions such that
\be\label{WatOsh}
\lim_{|\mathbf x|\to\infty}\langle \phi(\mathbf x,t) \phi(0,0)\rangle
\equiv
\lim_{|\mathbf x|\to\infty}\omega_0\left( \phi(\mathbf x,t) \phi(0,0) \right)  
= f(t)
\ee 
with $f$ a nontrivial periodic function of $t$.
However, if $\omega_0$ (is invariant under translations and) satisfies the cluster property, \emph{i.e.} defines a pure phase, one has (eq. \eqref{cluster_T})\footnote{The limit $\mathbf x \to \infty$ may be taken on a periodic subsequence $\mathbf x_n = n \mathbf a$, as needed in the case of space crystals, where the ground or equilibrium state is invariant under a subgroup $\mathcal T$ of space translations (see Section \ref{sec: Ground} and eq. \eqref{cluster_T}). }
\be\label{crit_WO}
\lim_{|\mathbf x|\to\infty}\langle \phi(\mathbf x,t) \phi(0,0) \rangle
=
\lim_{|\mathbf x|\to\infty}\langle \phi(\mathbf x,t)\rangle
\langle \phi(0,0)\rangle=
\langle \phi(0,t)\rangle
\langle \phi(0,0)\rangle
\,.
\ee
Then $f(t)=\langle \phi(0, t) \rangle\langle \phi(0, 0) \rangle$, which is a constant. Therefore, \eqref{WatOsh} may be effective only if $\omega_0$ violates the cluster property.

Under these general conditions, the argument actually holds for any internal symmetry 
$\beta_\lambda$: 
\be\label{invbeta}
\lim_{|\mathbf x|\to\infty}\langle \beta_\lambda \left(A(\mathbf x)\right) B\rangle = \langle \beta_\lambda(A)\rangle \langle  B\rangle=f(\lambda)\,,
\ee 
and the criterion is equivalent to the standard condition $\langle\beta_\lambda(A)\rangle\neq \langle A\rangle$.

A more general criterion has been proposed in \cite{Khemani}, which in the case of time translations includes the one proposed for Floquet systems \cite{Else}. For an internal symmetry $\beta_\lambda$, one considers states $|n\rangle$ which are eigenstates of the dynamics and of $\beta_\lambda$. The criterion identifies the spontaneous breaking of $\beta_\lambda$ with the existence of operators $\Phi_{i,\alpha}$ and $\Phi_{i,\bar\alpha}$ ($i$ being the site label) transforming as 
\be\label{irreptheta}
\beta_\lambda(\Phi_{i,\alpha})=e^{i\theta_\alpha(\lambda)}\Phi_{i,\alpha}
\ee with $\theta_{\bar\alpha}=-\theta_{\alpha}$ and $\theta_\alpha(\lambda)\neq0$ (implying $\langle n|\Phi_{i,\alpha}|n\rangle=0$) and such that
\be
\lim_{|i-j|\to\infty}\langle n|\Phi_{i,\alpha}\Phi_{j,\bar\alpha}|n\rangle_c \neq 0\,,
\ee
where $\langle\rangle_c$ denotes the connected part.

If $|n\rangle$ is invariant under (a suitable subgroup of) translations, its decomposition into translationally invariant states $\omega_\mu$ obeying the cluster property \cite{BratRobI} implies the existence of at least one $\omega_\mu$ which satisfies our criterion \eqref{eq: crit_inv}.
In fact, performing this decomposition gives ($c_{n,\mu}$ are positive coefficients with $\int d\mu\, c_{n,\mu}=1$)
\be
\langle n| \Phi_{i,\alpha} \Phi_{j,\bar\alpha} |n\rangle 
= \int d\mu\, c_{n,\mu}\, \omega_\mu(\Phi_{i,\alpha} \Phi_{j,\bar\alpha})
\xrightarrow[|i-j|\to\infty]{} 
\int d\alpha\, c_{n,\mu}\, \omega_\mu(\Phi_{i,\alpha}) \omega_\mu(\Phi_{j,\bar\alpha})\,,
\ee
where we have used the cluster property for $\omega_\mu$,
and the last expression can be nonzero only if $\omega_\mu(\Phi_{i,\alpha})\neq0$ and $\omega_\mu(\Phi_{j,\bar \alpha})\neq0$ for at least one $\alpha$. Then, Proposition \ref{criterio_t} applies. Thus the above criterion only guarantees that $\beta_\lambda$ is not a Wigner symmetry in at least one pure phase occurring in the pure-state decomposition of $|n\rangle$, without identifying this phase directly.

On the other hand, our criterion, Proposition \ref{criterio_t}, which applies also to the case of time translations, does not require to use operators satisfying \eqref{irreptheta}, whose existence is guaranteed if $\beta_\lambda$ is compact or if the symmetry acts as a product of single-site terms. Neither condition holds for time translations, which typically induce space delocalization, as also recognized by \cite{Khemani}. Moreover, Proposition \ref{criterio_t} directly identifies the phase where the symmetry is not a Wigner symmetry and may be more suitable from an operational point of view.

Another strategy, which has been at the basis of recent developments on the realization of quantum time crystals, is that employed by Floquet systems, proposed by \cite{Else} (for a recent review, see \cite{timecrystalsREVIEW}). In this setup the dynamics $\alpha_t^{(T)}$ is defined on the observables by the formal action of the unitary operator 
\be
\begin{cases}
	U_1(t) = e^{-i t H_1} &0<t<t_1\\
	U_2(t) = e^{-i (t-t_1) H_2}U_1(t_1) &t_1<t<t_1+t_2=T\,,
\end{cases}
\ee
with periodicity $T$. Clearly this means that the system is not isolated, being driven by an external periodic interaction. 
As stressed above, in order to break the symmetry of such a dynamics, one has to consider a pure phase which does not contain states invariant under the Floquet dynamics, equivalently one may employ Floquet eigenstates that violate the cluster property. 

Another advantage of our criterion is that it applies also to systems with time-independent Hamiltonians, with no need to introduce an external periodic drive and therefore to violate the one-parameter group property of the dynamics.

For the above reasons, in order to break time translations, we are led to consider pure phases characterized by translationally invariant states which are not ``eigenstates''of the full dynamics, in the thermodynamic limit, according to Proposition \ref{criterio_t}.

For constructing states which are not invariant under $\alpha_t$, a possible practical way is to consider a quantum system with a dynamics defined by cutoff Hamiltonians $H_V=H_{1,V}+H_{2,V}$, with $V$ indexing the cutoff (volume and/or number of constituents), and a corresponding state $\omega$, \emph{invariant under a group of translations $\mathcal T$}, which is invariant under the ``reference'' dynamics $\alpha_{1,t}$
\be\label{infinite-volume}
\alpha_{1,t} (A)=\lim_{V\to\infty}e^{it H_{1,V}} A e^{-it H_{1,V}}\,,\text{ for all }A\in\Alg\,,
\ee
but not under $\alpha_t$
\be
\alpha_t(A)=\lim_{V\to\infty}e^{it H_V} A\, e^{-it H_V}\,,\text{ for all }A \in \Alg\,.
\ee 
Clearly, from a practical point of view, the simpler the action of $H_2$ on $\omega$, the easier the analysis of the model. }

{
This setting goes beyond the standard quenched models since the invariance under translations and the thermodynamic limit play a crucial role.}

\section{Models of Broken Time Translations. \\Time Crystals.}
\label{sec: models}

We discuss concrete physical models that exhibit the spontaneous breaking of time translations. In some cases, the breaking occurs with residual periodic unbroken symmetry, realizing the basic characteristic features of quantum time crystals.

\subsection{Residual periodicity due to topology}\label{sec: particleonacircle}

As remarked in Section \ref{sec: Ground}, the existence of different pure phases may occur also for finite systems if the algebra of observables contains a nontrivial center, so that it is not described by the standard canonical (Weyl) algebra.

A prototypic example is given by a particle with charge $q$ and unit mass confined to a ring of unit radius that is threaded by (magnetic) flux $2\alpha/q$.
When $\alpha=0$ the model becomes the familiar quantum mechanical
model of a particle on a ring. In the textbook presentations, it is 
not sufficiently emphasized that the nontrivial topology of the circle
has a strong impact on the identification of the algebra of observables  
$\Alg$, which is restricted to be the subalgebra of the Weyl algebra invariant under rotations of $2\pi$, and therefore generated by 
$U(n)=e^{in\varphi}$ for $n\in\mathbb Z$ and by $V(\beta)=e^{i\beta p}$, where $p$ denotes the (angular) momentum \cite{Strocchi2015GAUGE}.

The observable $V(2\pi)=e^{i2\pi p}$ generates rotations of $2\pi$, and therefore belongs to the \emph{center} of $\Alg$. Its spectrum, \emph{i.e.} the expectation $\langle e^{i2\pi p}\rangle_\theta=e^{i\theta}$ for $\theta \in [0,2\pi)$, labels  irreducible representations of the observable algebra, which are hence called $\theta$-\emph{sectors}.
	The corresponding Hilbert-space representation is given either by $p= -i\partial/\partial \varphi + \theta/2\pi$ acting on periodic functions or, equivalently, by	$p=-i\partial/\partial\varphi$ acting on pseudo-periodic functions such that $\psi(\theta)=e^{i\theta}\psi(0)$.

An important observations is that momentum shifts $p\mapsto p + \alpha/2\pi$ given by
\be\begin{aligned}
	\rho^\alpha:\ &U(n)\mapsto U(n),\\
	&V(\beta)\mapsto V(\beta)e^{i\alpha/2\pi}
\end{aligned}\ee
define a symmetry of $\Alg$ which is broken in each $\theta$-sector unless $\alpha$ is an integer multiple of $2\pi$, the reason being that $\rho^\alpha$ does not commute with the center:
\be
\langle \rho^\alpha\left(e^{i2\pi p}\right)\rangle_\theta = e^{i(\theta+\alpha)}=\langle e^{i2\pi p}\rangle_{\theta+\alpha}.
\ee 
Let us mention that this situation is strikingly similar to that of QCD with $V(2\pi n)$ playing the role of large gauge transformation and $\rho^\alpha$ of chiral transformation \cite{FS_M, FS_QFT} and that this models spontaneously breaks parity $\Omega: p\to-p$, $\varphi\to-\varphi$ unless $\theta=0,\pi$.

Now, let us denote by $\alpha_t^\alpha$ the dynamics in the presence of magnetic flux, \emph{i.e.} generated by the Hamiltonian
\be
H^\alpha = \frac{1}{2}(p-\alpha)^2\,,
\ee
which can be expressed in terms of the dynamics at vanishing magnetic flux $\alpha^0_t$ and the momentum shift symmetry as follows
\be 
 \alpha^\alpha_t = \rho^{-2\pi\alpha}  \alpha^{0}_t  \rho^{2\pi\alpha}.
\ee
The $\theta$-sectors correspond to the factorial representations defined by eigenstates of $\alpha^\alpha_t$.
This dynamics plays the role of the reference dynamics (the analog of that generated by $H_1$ in the previous general discussion) with eigenstates labeled by $\theta$.

In order to follow the general strategy outlined above, we now introduce a perturbation of $\alpha^\alpha_t$, leading to the new dynamics $\alpha^{\alpha,k}_t$ such that a given $\theta$-sector does not contain eigenstates of  $\alpha^{\alpha,k}_t$. A simple example fulfilling such requirements is obtained by considering  $\alpha_t^{\alpha,k}$ as the automorphism formally defined by 
\be
 H^{\alpha, k} \equiv H^{\alpha} - k\, \varphi\,.
\ee
For simplicity we consider the case $\alpha=0$, the extension to nonzero $\alpha$ being straightforward.

The a Hamiltonian $H^{0,k}$ does not
belong to the algebra of observables, but nevertheless the corresponding
time evolution $\alpha^{k}_t$ maps the algebra of observables $\Alg$ into itself,
as required for physical consistency. In fact, using Zassenhaus' formula \cite{FS_SSB},
\be\begin{aligned}
	\alpha^{k}_t \left(U(n)\right) &= e^{	\frac{it}{2} p^2}U(n)e^{-\frac{it}{2} p^2} e^{\frac{it^2}{2}nk}
	=U(n) V(nt) e^{int(n+kt)/2}
	\in\Alg\\
	\alpha^{k}_t \left(V(\beta)\right) &=
	V(\beta) e^{i\beta k t}\,.
\end{aligned}\ee
This dynamics corresponds to $p\mapsto p + kt$ (\emph{i.e.} to a constant acceleration) and, formally, to $\varphi\mapsto \varphi + pt + kt^2/2$\,.   
The crucial fact is that such a dynamics does not leave the center pointwise invariant:
\be
\alpha^{k}_t \left(V(2\pi)\right) = V(2\pi) e^{i2\pi k t}.
\ee
This equation also holds for $\alpha_t^{\alpha, k}$ for any $\alpha$,
and implies that time translation symmetry is spontaneously broken in each $\theta$-sector, leaving as a residual unbroken group the discrete subgroup corresponding to 
\be
t_n = \frac{n}{k}\qquad \text{for }n\in\mathbb Z.
\ee 

Restricting the discussion to irreducible/factorial representations is of course crucial for obtaining symmetry breaking in this finite-dimensional model: in the (reducible) representation with a decompactified variable $q$ such that $\varphi= q \text{ mod }2\pi$, given by $L^2(\mathbb R, dq)$, time evolution is indeed implemented by a unitary operator which, however, does not leave each $\theta$-sector invariant.

Albeit very simple, the model proves that one may (spontaneously) break time translation symmetry with residual periodic group. This proves that the characteristic features may be realized, and no apriori obstruction exists. 
The nontrivial topology provides the crucial mechanisms of the instability of the center of observables under time translations, and of the residual periodic structure.

\subsection{Quantum time crystal realized by a spin model}
As discussed in Section \ref{sec: constructive!}, the thermodynamic limit may play a crucial role both in breaking time translations and in leaving residual periodic symmetry.
This is clearly displayed by the following spin model.
 
We start by considering a lattice spin one-half system, whose (local) observable algebra $\Alg$ is generated by localized polynomials of the spins $\mathbf{s}^i$, where $i$ denotes the lattice site. We introduce a dynamics defined by spin one-half cutoff Hamiltonians 
$$
H_{V}=H_{1,V}+H_{2,V}\,;
$$
with $H_{1,V}$
invariant under lattice translations and under (spin) rotations, and 
\be
H_{2,V}= - h \sum_{i\in V}s^i_z\,,
\ee 
corresponding to the interaction with a uniform magnetic field $\mathbf h = h \hat{\mathbf z}$ pointing in the $z$ direction.

We consider a translation-invariant state $\omega$, invariant under the dynamics $\alpha_{1,t}$, generated by $H_{1,V}$ in the thermodynamic limit (according to \eqref{infinite-volume}), satisfying the cluster property, and such that:
\be
\omega(s_x^i)\neq0\,.
\ee

The rotational invariance of $H_{1,V}$ implies that it commutes with $\alpha_{2,t}$, generated by $H_{2,V}$ in the thermodynamic limit, so that $\alpha_t = \alpha_{1,t}\alpha_{2,t}$. Then, one has
\bea
\omega(\alpha_t(s^i_x))  = \omega(\alpha_{1,t}\alpha_{2,t}(s^i_x)) 
&=\omega(\alpha_{2,t}(s_x^i))\\
&=\omega(s_x^i) \cos (ht) - \omega(s_y^i) \sin (ht)\,.
\eea

Thus, by Proposition \ref{criterio_t}, time translation symmetry $\alpha_t$ is spontaneously broken in the phase $\Gamma_\omega$ defined by $\omega$, with residual unbroken subgroup $\alpha_{t_n}$  
\be
t_n = \frac{2\pi n}{h}\qquad \text{for $n\in\mathbb Z$}\,.
\ee 
The simplest realization of such a structure is obtained when $H_{1,V}$ is the nearest-neighbour ferromagnetic Heisenberg  Hamiltonian and $\omega$ is the (factorized) state with all spins pointing in the $x$ direction (see \cite{FSCH_QTC}; this model was more recently considered in \cite{HeisenbergQTC}).

As in the case of Floquet time crystals, the eigenstates of $\alpha_t$ do not belong to $\Gamma_\omega$, {and} this is achieved without the introduction of an external periodic drive and therefore without spoiling the one-parameter group property of the dynamics. 

The non-invariance of $\omega$ under $\alpha_t$ ensures the breaking of time translations, in the phase $\Gamma_\omega$ defined by it, because, thanks to the thermodynamic limit and the cluster property, local perturbations of $\omega$ cannot induce a transition to the state $\omega_z$ with all the spins pointing in the $z$ direction (which, therefore, defines a disjoint phase).
 
The model is interesting for the breaking of time translations only in the thermodynamic limit, since no breaking occurs for finite volumes (or for a finite number of spins).

{
It is instructive to retrace the strategy proposed in \cite{Khemani} starting from our state $\omega$. To this purpose, one may define the state
\be
\Omega(A)=\int_0^{2\pi}\frac{d\varphi}{2\pi}\,\omega\left(\beta_{\varphi}(A)\right)\,,
\ee
where $\beta_\varphi$ denotes spin rotations about the $z$ axis. By construction $\Omega$ is invariant under $\alpha_{1,t}$ and $\alpha_{2,t}$, and hence under $\alpha_t$. Considering then $s_\pm^i=s_x^i\pm i s_y^i$, we have $\beta_\varphi(s_\pm^i)=e^{\pm i\varphi}s_\pm^i$ and 
\be
\Omega(s^i_+ s^j_-)
=\omega(s^i_+ s^j_-)
\xrightarrow[|i-j|\to\infty]{}
\omega(s^i_+)\omega(s^j_-)=
\omega(s^i_x)^2+\omega(s^i_y)^2\,,
\ee 
since $\omega$ satisfies the cluster property.
The latter equation highlights that $\Omega$ violates the cluster property if and only if either $\omega(s_x^i)\neq0$ or $\omega(s_y^i)\neq0$, \emph{i.e.} $\omega$ breaks $\alpha_{2,t}$.
}

In agreement with the general characterization of spontaneous symmetry breaking (Section \ref{sec: stbb}), the spontaneous breaking of time translations forbids the existence of a one-parameter group $U(t)$ of unitary operators implementing the time translations $\alpha_t$ in {$\Gamma_{\omega}$. Therefore} it is impossible to define a Hamiltonian, in the thermodynamic limit (only the finite-volume Hamiltonian $H_V$ may be defined). More explicitly, the finite-volume dynamics converges only as an automorphism of the algebra $\Alg$, whereas $\exp\{-it(H_{1,V}+H_{2,V})\}$ does not converge as $V\to\infty$, even if one introduces subtractions and/or counterterms. This further shows the dramatic difference with respect to finite-dimensional systems, like the two level model, which bears no physical relevance for the possibility of breaking time translations.
 
As explicitly realized by the model, in agreement with Proposition \ref{Criterio_centro}, the breaking of time translations with a discrete residual subgroup is immediately displayed once the average observables undergo a periodic motion.

\subsection{An improved Wilczek model}
In order to support the possible existence of quantum time crystals, Wilczek proposed the following many-particle model, consisting of $N$ particles on a ring with unit radius, labeled by the index $i$, with the Hamiltonian: 
\be
H= \sum_{i=1}^N H^\alpha_i - \frac{\lambda}{N-1}\sum_{i\neq j}\delta(\varphi_i-\varphi_j)\,.
\ee
According to the strategy discussed above, with the aim of considering a thermodynamic limit, one has to start by defining the algebra of observables. Since each particle lives on the ring, it is natural to consider, as the $i$th-particle observable algebra, the algebra $\Alg_i$ generated by $U(n_i)=e^{in_i \varphi_i}$, for $n_i\in\mathbb Z$, and $V(\beta_i)=e^{i\beta_i p_i}$\, (\emph{i.e.} the algebra discussed in \ref{sec: particleonacircle}). Then, the $N$-particle algebra is generated by products of the single-particle algebras. 
We consider the model defined by the following Hamiltonian: 
\be
H^{\alpha,\mathcal V,N} = \sum_{i=1}^N
\big[H^{\alpha}_i+\mathcal V_i(\varphi_i)\big]+\sum_{i<j} v(\varphi_i -\varphi_j)
\ee
with $\mathcal V_i(\varphi)$ and $v(\varphi)$ periodic functions, to guarantee that the corresponding time evolution $\alpha^{\alpha,\mathcal V,N}_t$ maps
observables into observables.

{
We remark that the Hamiltonians at the basis of Wilczek's model \cite{Wilczek} and of \cite{Wilczek-excited} do not use periodic potentials and do not guarantee this condition, whose violation is in conflict with the physical interpretation of particles on a circle.
}
	
In order to find phases with broken time translations in the $N\to\infty$
limit, consider the product state $\Psi_0^N \equiv \prod \Psi_0^i$, with $\Psi_0^i(\varphi_i)=(2\pi)^{-1/2}e^{i\alpha\varphi_i}$
the ground state of $H_i^{\alpha}(\varphi_i)$ (one might as well take the product of the same
$\theta$-state for each value of the index $i$). Clearly, in the thermodynamic limit such
a state defines a pure phase and is invariant under the group of
``translations'' of the index $i$, so that the criterion of Section \ref{sec: Ground} applies.

To this purpose, consider the observable
\be
P^f_N= \frac{1}{N} \sum_{i=1}^N f_i(\varphi_i)p_i f_i(\varphi_i) \equiv \frac{1}{N} \sum_{i=1}^N p^f_i,
\ee
where $f_i$ is a smooth (real) function of compact support in $[0,2\pi)$ with periodic extension $f_i(\varphi_i+2\pi)=f_i(\varphi_i)$. The observable $p_i^f$ describes a momentum localized in the support of $f_i$ and its expectation $\langle p_i^f\rangle_i$ on the corresponding state $\Psi_0^i$ gives
\be\begin{aligned}
\langle p_i^f\rangle_i &= \int_0^{2\pi} e^{-i\alpha\varphi}f_i(\varphi) \left(-i\frac{\partial}{\partial \varphi} f_i(\varphi) e^{i\alpha\varphi}\right) \frac{d\varphi}{2\pi} \\
&=-i \int_0^{2\pi} \frac{\partial}{\partial \varphi}f_i^2(\varphi)\frac{d\varphi}{4\pi} + \alpha \int_0^{2\pi}f_i^2(\varphi)\frac{d\varphi}{2\pi},
\end{aligned}\ee
but the first term vanishes by the periodicity of $f_i$ and we are left with the non-vanishing result
\be
\langle p_i^f\rangle_i=\alpha \int_0^{2\pi}f_i^2(\varphi)\frac{d\varphi}{2\pi}.
\ee

Furthermore, denoting by $\alpha^{\alpha,\mathcal V,N}_t$ the time evolution generated by $H^{\alpha,\mathcal V,N}$, and by $\langle \cdot \rangle_{\alpha,N}$ the expectation on the above product state $\Psi_0^N$, we obtain
\be
\frac{d}{dt}\langle \alpha^{\alpha,\mathcal V,N}_t(p_i^f)\rangle_{\alpha,N}\Big|_{t=0}= 
\int_0^{2\pi} f_i^2(\varphi) \mathcal V_i'(\varphi)\frac{d\varphi}{2\pi},
\ee
using the fact that $\Psi_0^N$ is an eigenstate of the kinetic Hamiltonian, but not of  $H^{\alpha,\mathcal V,N}$. Note that the contributions due to $v(\varphi_i-\varphi_j)$ actually vanish.

The macroscopic observable
\be
P^f_\text{av} = \lim_{N\to\infty}P_N^f
\ee
as usual defines an element of the center of the observables 
and 
\be\label{eq: non_P}
\frac{d}{dt}\langle \alpha^{\alpha,\mathcal V,\infty}_t(P_\text{av}^f)\rangle_{\alpha,\infty}\Big|_{t=0}
=
\lim_{N\to\infty}\frac{1}{N}\sum_{i} \int_0^{2\pi} f_i^2(\varphi) \mathcal V_i'(\varphi_)\frac{d\varphi}{2\pi}\,;
\ee
whenever $\int f_i^2 \mathcal V'_i \,d\varphi$ admits nonvanishing Cesaro limit as $i\to\infty$, the average variable $P_\mathrm{av}$ is not invariant under time evolution, and this implies breaking of time translations.

Another possibility is to consider the average observable defined by 
\be
F_\text{av} = \lim_{N\to\infty}\frac{1}{N}\sum_{i=1}^NF_i(\varphi_i)
\ee
where $F_i(\varphi_i)$ is a smooth, periodic function of compact support. 
Proceeding as before, \emph{i.e.} exploiting the periodicity of $F_i$ to get rid of a few intermediate terms, we find
\be\begin{aligned}\label{eq: non_F}
\frac{d^2}{dt^2}\langle\alpha^{\alpha,\mathcal V,\infty}_t\left(F_\text{av}\right)\rangle_{\alpha,\infty}\Big|_{t=0}&=\lim_{N\to\infty}\frac{1}{N}\sum_{i}\int_0^{2\pi} F'_i(\varphi)\mathcal V_i'(\varphi)\frac{d\varphi}{2\pi}.
\end{aligned}\ee
Whenever $\int F'_i\mathcal V_i'\,d\varphi$ has nonzero Cesaro limit as $i\to\infty$,   
we have breaking of time translations.

\section*{Conclusions}
The main result of the paper is the characterization of spontaneous symmetry breaking in pure phases defined by a translationally invariant state in terms of its noninvariance under the symmetry.
Such a state need not be a ground or vacuum state and therefore {our criterion  substantially generalizes the standard one. In particular, it allows for the breaking of time translations. This is explicitly shown to work in explicit models without the need of introducing an ad hoc external periodic drive.} 

\section*{Acknowledgements}
F.S. is grateful to Giovanni Morchio for very useful discussions.


\providecommand{\href}[2]{#2}\begingroup\raggedright\endgroup

\end{document}